\documentclass[11pt,index]{amsart}

\usepackage{ulem,xspace} %,titletoc}

\usepackage[pagebackref,colorlinks,linkcolor=red,citecolor=blue,urlcolor=blue,hypertexnames=true]{hyperref}
\usepackage{comment}
\usepackage{ulem}
\usepackage{upgreek}
\usepackage{amsmath}
\usepackage{hyperref}
\usepackage{cleveref}
\usepackage{autonum}

\usepackage{amsthm,amsfonts,amsmath,amssymb,amscd, accents, mathrsfs, mathtools}
\usepackage{hyperref}
                                \usepackage{verbatim}
                                \usepackage{graphicx}

\swapnumbers
                                
\theoremstyle{plain}

\usepackage{enumitem}
\numberwithin{equation}{section}
\newtheorem{theorem}{Theorem}[section]

\newtheorem{lemma}[theorem]{Lemma}
\newtheorem{proposition}[theorem]{Proposition}

\theoremstyle{definition}
\newtheorem{definition}[theorem]{Definition}
\newtheorem{notation}[theorem]{Notation}
\newtheorem{example}[theorem]{Example}

\theoremstyle{remark}
\newtheorem{remark}[theorem]{Remark}
\numberwithin{equation}{section}
\newcommand{\R}{{\mathbb R}}

\newcommand{\cB}{{\mathcal B}}

\newcommand{\cR}{{\mathcal R}}

\newcommand{\cH}{{\mathcal H}}
\newcommand{\cL}{{\mathcal L}}
\newcommand{\cX}{{\mathcal X}}
\newcommand{\cC}{{\mathcal C}}

\newcommand{\cN}{{\mathcal N}}
\newcommand{\cM}{{\mathcal M}}
\newcommand{\cK}{{\mathcal K}}
\newcommand{\cY}{{\mathcal Y}}
\newcommand{\cI}{{\mathcal I}}

\newcommand{\cU}{{\mathcal U}}
\newcommand{\cV}{{\mathcal V}}
\newcommand{\cE}{{\mathcal E}}
\newcommand{\cD}{{\mathcal D}}

\newcommand{\ket}[1]{\left\vert #1\right\rangle}
\newcommand{\bra}[1]{\left\langle #1\right\vert}

\newcommand{\Tr}{\mathrm{Tr}}

\newcommand{\be}{\begin{equation}}
\newcommand{\ee}{\end{equation}}
\newcommand{\bea}{\begin{eqnarray}}
\newcommand{\eea}{\end{eqnarray}}
\newcommand{\beann}{\begin{eqnarray*}}
\newcommand{\eeann}{\end{eqnarray*}}

\usepackage{color}

\begin{document}

\title{Dynamical Coherence Measures}
 \author[Vershynina]{Anna Vershynina}
 \address{Department of Mathematics\\
 University of Houston \\ TX \\ USA}
 \email{anna@math.uh.edu}
% \thanks{${}^2$Supported by NSF grants DMS-1812734, DMS-2105583.  } 
%\renewcommand\Authands{ and }
%\renewcommand\Affilfont{\itshape\small}

\date{\today}


\begin{abstract}
We present three measures of the dynamical coherence of channels, which generalize several previous results. The measures are based on a generalized distance-like function between channels, such as a quantum divergence of channels, including the relative entropy and diamond distance. Free operations are considered to be detection-incoherent, creation-incoherent, and detection-creation incoherent. We consider free superchannels constructed from free pre- and post-processing channels. The new measures are monotone under free superchannels and convex.
\end{abstract}
\maketitle

\section{Introduction}

Quantum coherence describes the existence of quantum interference, and it is often used in thermodynamics \cite{A14, C15, L15}, transport theory \cite{RM09, WM13}, and quantum optics \cite{G63, SZ97}, among a few applications. Quantum coherence resource theory starts with free, i.e. incoherent, states, which are diagonal states in a pre-fixed basis. Free operations are some quantum channels that do not create coherence where it was absent, in other words, map the set of incoherent states to itself. Problems involving coherence include quantification of coherence \cite{BC14, RPL16, R16, SX15, V22, V23, YZ16}, distribution \cite{RPJ16}, entanglement \cite{CH16, SS15}, operational resource theory \cite{ChG16-2, CH16, DBG15, WY16},  and correlations \cite{HH18, MY16, TK16}. Operational aspects of coherence resource theory also include coherence distillation, with deterministic distillation under strictly incoherent operations studied in \cite{liu2019deterministic} and stochastic distillability under incoherent operations characterized in \cite{liu2022criterion}. See \cite{SAP17} for a more detailed review.

Relatively recently, static resource theories (i.e. the ones mostly concerned with states and their manipulation) have been extended to regard quantum channels as the elementary generalized resource, leading to an active area of research in dynamical resource theory. Static resource theory has three main components: free states, free operations and resource measures. Analogously, dynamical resource theory must have free channels, free superchannels and resource measures. Note that the dynamical theory is a generalization of the static one since any state can be regarded as a quantum channel mapping a trivial state to a given one. Much work has focused on the development of the theory of entropic quantifiers of channels and operational resource theory \cite{ChDAP08, GW21, LY20, SCG20, TEZP19, Y18}, to name a few references.

In \cite{Y18} it was shown that when the distance is measured by the relative entropy of channels (see Definition \ref{def:div}), the minimal distance to the free channels is a coherence measure of channels (see Definition \ref{def:coh-meas-dyn}). Denote by $\cC$ a set of free channels, taken to be either detection-incoherent (DI), creation-incoherent (CI), or detection-creation incoherent (DCI) (see Definition \ref{def:DCI}). Then it was shown that for a relative entropy of channels $D_{re}$, 
\begin{equation}
C^{\cC}(\cN)=\min_{\cK\in \cC}D_{re}(\cN,\cK)
\end{equation}
is a coherence measure of channels. This measure is monotone for superchannels with free pre- and post-processing channels (i.e. of the form (\ref{eq:free-super-one})). 

In \cite{TEZP19} the authors considered a norm $\|\cdot\|$ that is sub-multiplicative on quantum channels, sub-multiplicative with respect to the tensor products, and such that $\|\cK\|\leq 1$ for any detection-incoherent $\cK$ (e.g. a diamond norm, see Definition \ref{def:diamond-norm}), and showed that the following functional
\begin{equation}\label{eq:norm-intro}
C^{DI}(\cN)=\min_{\cK\in DI}\|\Delta\cN-\Delta\cK \| \ 
\end{equation}
is a coherence measure of channels. Here $\Delta$ is a completely dephasing channel on an output system (see Notation \ref{notation:cdo}). Additionally, it was shown that  (\ref{eq:norm-intro}) is a coherence measure when $\|\cdot\|$ is a trace-one norm (see Definition \ref{def:trace-norm}). The operational interpretation of these measures is that in a one-shot regime they appear in the probability to guess correctly if one obtained the target channel or the least distinguishable free operation, provided we can use only free measurements for the trace-one norm, and  non-free measurements are allowed on the auxiliary system for the diamond norm.

Motivated by these distance-based constructions, we show that functions $f$ satisfying the properties specified in Notation \ref{def:f} generate the following coherence measures:
\begin{align}
&C_f^{DI}(\cN)=\inf_{\cK\in DI}f(\Delta\cN,\Delta\cK),\label{def:DI-f-intro}\\
&C_f^{CI}(\cN)=\inf_{\cK\in CI}f(\cN\Delta,\cK\Delta),\label{def:CI-f-intro}\\
&C_f^{DCI}(\cN)=\inf_{\cK\in DCI}f(\cN,\cK). \label{def:DCI-f-intro}
\end{align}
Here $f$ is a distance-like function between channels satisfying a particular list of properties, see Notation \ref{def:f}. As examples, $f$ can be taken to be a quantum divergence of channels, including the relative entropy of channels and the diamond distance. The dephasing channel acts on the output system in (\ref{def:DI-f-intro}) and the input system in (\ref{def:CI-f-intro}) (see Definition \ref{def:DCI}). The superscript specifies the chosen class of free operations: detection-incoherent (DI), creation-incoherent (CI), or detection-creation incoherent (DCI) (see Definition \ref{def:DCI}). The subscript $f$ specifies the function used.

The results apply to finite-dimensional systems and to the class of free superchannels constructed from free pre- and post-processing channels considered here.

In Section \ref{sec:Prelim}, we start with common notions and notations: relative entropy/divergence of channels, static and dynamical coherence theories. 

In Section \ref{sec:Coh-meas}, we introduce a general class of distance-like functions between quantum channels and show that they generate convex dynamical coherence measures (\ref{def:DI-f-intro}), (\ref{def:CI-f-intro}), (\ref{def:DCI-f-intro})  for the DI, CI, and DCI classes, monotone under the corresponding free superchannels. We also discuss faithfulness and the interpretation of these constructions as distances to the free sets.

%In Section \ref{sec:Examples}, we present explicit qubit examples illustrating the three coherence measures and showing how the DI, CI, and DCI constructions distinguish different aspects of dynamical coherence.

\section{Preliminaries}\label{sec:Prelim}

In this paper, we use the following notations and definitions: all Hilbert spaces are finite dimensional. We denote them as $\cH_A, \cH_B, \dots$, where subscripts indicate the corresponding system. Sometimes we drop the subscript if the statement is true for any Hilbert space or it is evident which system we are talking about. For a Hilbert space $\cH$, we denote by $\cL(\cH)$ the space of all linear maps on $\cH$, and by $\cB(\cH)$ the space of all bounded linear operators on $\cH$.

A Hilbert space corresponding to multiple systems, e.g. AB, is a tensor product of individual subsystems, e.g. $\cH_A\otimes\cH_B$. 

For a Hilbert space $\cH_A$, its dimension is denoted as $|A|:=\dim\cH_A$. A quantum state or a density operator $\rho_A\in\cB(\cH_A)$ on a Hilbert space $\cH_A$ is a positive semidefinite, trace-normalized operator, i.e. $\rho_A\geq 0$ and $\Tr\,\rho_A=1$. A state is pure if it is rank-one. A pure state $\psi_A$ has an associated  vector $\ket{\psi}_A\in\cH_A$ such that $\langle{\psi}|\psi\rangle=1$ and $\psi_A=\ket{\psi}\bra{\psi}_A$. The set of all density operators or states on $\cH_A$ is denoted as $\cD(\cH_A)$. For a generic Hilbert space $\cH$, we denote $\cD=\cD(\cH)$.
 
A quantum channel $\cN:A\rightarrow B$ is a linear completely-positive trace-preserving (CPTP) map from $\cL(\cH_A)$ to $\cL(\cH_B)$. Sometimes the channel can also be denoted as $\cN_{A\rightarrow B}$. The identity channel on $\cL(\cH_A)$ is denoted as $I_A$. The subscript in the channels is dropped if it is evident which systems are involved. If systems $A$ and $B$ are the same, then the channel $\cN_{A\rightarrow A}$ is denoted as $\cN_A$. A composition of channels $\cN:A\rightarrow B$ and $\cK:B\rightarrow C$ is denoted as $\cK\circ\cN$, although sometimes we use notation $\cK\cN$ to simplify the formula.

A {superchannel} \cite{ChDAP08} $\Lambda$ transforms a quantum channel $\cN_{A\rightarrow B} $ to a channel from $C$ to $D$ as follows
\begin{equation}\label{def:superch}
\Lambda(\cN_{A\rightarrow B})_{C\rightarrow D}=\cK_{BE\rightarrow D}\circ(\cN_{A\rightarrow B}\otimes I_E)\circ \cM_{C\rightarrow AE}\ , 
\end{equation}
with the ancillary system $E$, and channels $\cK_{BE\rightarrow D}$ and $\cM_{C\rightarrow AE}$.

\subsection{Divergence of channels}

We consider the following class of divergences, which serve as distinguishability measures between quantum states.

\begin{definition}\label{def:div-div} A function $D:\cD\times\cD\rightarrow \R\cup\{+\infty\}$ on a set of pairs of states is a (generalized convex) {\it divergence} if
\begin{itemize}
\item (Non-negativity) $D(\rho\|\sigma)\geq 0$, and $D(\rho\|\sigma)=0$ if and only if $\rho=\sigma$.
\item (Data processing/monotonicity). For any quantum channel $\cN$, it holds that $D(\cN(\rho)\|\cN(\sigma))\leq D(\rho\|\sigma)$.
\item (Stability) $D(\rho\otimes\tau\|\sigma\otimes\tau)=D(\rho\|\sigma)$.
\item (Joint convexity) For $0\leq \lambda\leq 1$, and quantum states $\rho_i, \sigma_i$, it holds that
\begin{align}
&D\left(\lambda\rho_1+(1-\lambda)\rho_2\|\lambda\sigma_1+(1-\lambda)\sigma_2\right)\leq \lambda D(\rho_1\|\sigma_1)+(1-\lambda)D(\rho_2\|\sigma_2)\ . 
\end{align}
\end{itemize}
\end{definition}
Note that even though we explicitly state the stability property, it is a consequence of the monotonicity property, since both tensoring and the partial trace are CPTP maps (i.e. quantum channels) \cite{WWY14}. Moreover, from the monotonicity property we obtain the invariance under unitaries, i.e. $D(\rho\|\sigma)=D(U\rho U^*\|U\sigma U^*)$ for any unitary $U$.

\begin{example}
(Umegaki) Quantum relative entropy, $D_{re}(\rho\|\sigma)=\Tr(\rho\log\rho-\rho\log\sigma)$, and trace-distance, $\|\rho-\sigma\|_1=\Tr|\rho-\sigma|$, are examples of a generalized convex divergence. 
\end{example}

\begin{definition}\label{def:div}
Let $D$ be a divergence from Definition \ref{def:div-div}. Then a {\it quantum divergence of channels} $\cN_{A\rightarrow B}$ and $\mathcal{M}_{A\rightarrow B}$ based on $D$ is defined as 
\begin{equation}\label{def:qre-channel}
D(\cN \| \cM)=\sup_{\rho_{AR}} D(\cN\otimes I(\rho)\| \cM\otimes I(\rho))\ .
\end{equation}
Here the supremum is taken over all systems $R$ of any dimension and all states $\rho_{AR}$. However, it is sufficient to consider only pure states $\rho_{AR}$ with system $R$ being isomorphic to system $A$, because of the state purification, the data-processing inequality, and the Schmidt decomposition theorem. 
\end{definition}

Taking Umegaki relative entropy $D_{re}$ in the definition of the quantum divergence of channels, it becomes a relative entropy of channels, which was first proposed in \cite{CMM18}, and generalized in \cite{LKDW18}. 

When divergence is a trace-distance, then the divergence of channels is called a {\it diamond-distance} of channels, which we discuss below.

The quantum divergence of channels satisfies the following properties, which directly follow from the proofs in \cite{Y18}:
\begin{itemize}
\item (Non-negativity) $D(\cN\|\cM)\geq 0$ and $D(\cN\|\cM)=0$ if and only if $\cN=\cM$.
\item (Weak monotonicity) For any quantum channels $\mathcal{K}_i$, we have 
$$D(\mathcal{K}_1\circ\cN\circ \mathcal{K}_2\|\mathcal{K}_1\circ\cM\circ \mathcal{K}_2 )\leq D(\cN\|\cM)\ .$$
\item (Strong monotonicity) For any superchannel $\Lambda$, 
$$D(\Lambda(\cN)\|\Lambda(\cM))\leq D(\cN\|\cM)\ .$$
\item (Joint convexity) For $0\leq\lambda\leq 1$, and quantum channels $\cN_i, \cM_i$, we have 
\begin{align}
&D\left(\lambda\cN_1+(1-\lambda)\cN_2\|\lambda\cM_1+(1-\lambda)\cM_2\right)\leq \lambda D(\cN_1\|\cM_1)+(1-\lambda)D(\cN_2\|\cM_2)\ .
\end{align}
%\item (Additivity) $D(\cN_1\otimes\cN_2\|\cM_1\otimes\cM_2)\geq D(\cN_1\|\cM_1)+D(\cN_2\|\cM_2)$.
\item (Stability) $D( \cN\otimes I\| \cM\otimes I)=D(\cN\|\cM)$.
 \end{itemize}
 The properties of the quantum divergence of channels are discussed in Section \ref{sec:Coh-meas}.

\subsection{Trace and diamond norms}

\begin{definition}\label{def:trace-norm} \cite{Wat18} The {\it trace-norm of a linear map} $\Omega:\cL(\cH_A)\rightarrow\cL(\cH_B)$ is defined as 
\begin{equation}\label{def:trace-norm-channel}
\|\Omega\|_1=\max \{ \|\Omega(X)\|_1\ :\ X\in\cL(\cH_A),\ \|X\|_1\leq 1\}  \ ,
\end{equation}
where the trace-norm on $\cL(\cH)$ is defined as $\|M\|_1=\Tr\sqrt{M^*M}$ for $M\in\cL(\cH)$. Note that the maximum above is always achieved by a rank-one operator (Proposition 3.38 \cite{Wat18}). For a positive map the maximum is achieved by a rank-one projection.
\end{definition}
\begin{definition}
The {\it trace-distance of quantum channels} $\cN, \cM:A \rightarrow B$ is defined as 
\begin{equation}\label{def:trace-norm-channel}
\|\cN -\cM\|_1=\max_{\rho\in\cD(\cH_A)} \|\cN(\rho) - \cM (\rho)\|_1\ .
\end{equation}
\end{definition}

The trace-norm of a quantum channel has the following properties (for more properties see \cite{Wat18}):
\begin{enumerate}
\item The trace-norm of a quantum channel $\cN$ is one, i.e. $\|\cN\|_1=1$.
\item For all linear maps $\cX: \cL(\cH_A)\rightarrow\cL(\cH_B)$, $\cY: \cL(\cH_B)\rightarrow\cL(\cH_C)$, it holds $\|\cY\cX\|_1\leq \|\cY\|_1\|\cX\|_1$. 
\item For all channels $\cN_1,\cM_1:A\rightarrow B$ and $\cN_2, \cM_2:B\rightarrow C$, it holds that
$$\|\cN_2\circ\cN_1-\cM_2\circ\cM_1 \|_1\leq \|\cN_1-\cM_1\|_1+\|\cN_2-\cM_2\|_1\ .$$
\end{enumerate}

\begin{definition}\label{def:diamond-norm} \cite{Wat18}
The {\it diamond norm (or the completely bounded trace-norm) of a linear map} $\Omega: \cL(\cH_A)\rightarrow\cL(\cH_B)$ is defined as 
$$\|\Omega_{A\rightarrow B}\|_\diamond=\|\Omega_{A\rightarrow B}\otimes I_R\|_1\ , $$
where $\dim R=\dim A$.
%It is sufficient to consider only systems $R$ isomorphic to the system $A$ .  
\end{definition}

In particular, the diamond norm satisfies the following properties \cite{Wat18}:
\begin{enumerate}
\item The diamond norm of a quantum channel $\cN$ is one, i.e. $\|\cN\|_\diamond=1$. 
\item Since trace-norm is monotone under quantum channels, in particular partial traces, we get that for any linear map $\Omega$, 
$$ \|\Omega\|_1\leq \|\Omega\|_\diamond\ . $$
\item (Sub-multiplicativity) for any linear maps $\Omega: \cL(\cH_B)\rightarrow\cL(\cH_C)$ and $\Sigma: \cL(\cH_A)\rightarrow\cL(\cH_B)$
$$\|\Omega\circ\Sigma\|_\diamond\leq \|\Omega\|_\diamond\|\Sigma\|_\diamond\ . $$
In particular, composition with a unitary channel on either side leaves the diamond norm invariant: for any linear map \(\Omega\) and unitary channels \(\mathcal U,\mathcal V\) of compatible dimensions,
\[
\|\mathcal U\circ\Omega\circ\mathcal V\|_\diamond=\|\Omega\|_\diamond.
\]
Indeed, one inequality follows from sub-multiplicativity, and the reverse inequality follows by composing with the inverse unitary channels.
\item For all channels $\cN_1,\cM_1:A\rightarrow B$ and $\cN_2, \cM_2:B\rightarrow C$, it holds that
$$\|\cN_2\circ\cN_1-\cM_2\circ\cM_1 \|_\diamond\leq \|\cN_1-\cM_1\|_\diamond+\|\cN_2-\cM_2\|_\diamond\ .$$
\item (Multiplicativity under tensor products) for any linear maps $\Omega: \cL(\cH_A)\rightarrow\cL(\cH_B)$ and $\Sigma: \cL(\cH_C)\rightarrow\cL(\cH_D)$
$$\|\Omega\otimes\Sigma\|_\diamond= \|\Omega\|_\diamond\|\Sigma\|_\diamond\ . $$
\item (Monotonicity under superchannels)  Let $\Lambda$ be a superchannel (\ref{def:superch}). Then from the above properties it follows that
$$\|\Lambda(\cN)\|_\diamond\leq \|\cM\|_\diamond \|\cN\otimes I\|_\diamond\|\cK\|_\diamond=\|\cN\|_\diamond \ ,$$
since $\cM$ and $\cK$ are channels,  their diamond norm is one, and since $\|\cN\otimes I\|_\diamond=\|\cN\|_\diamond$.
\end{enumerate}

\begin{definition}
Taking the trace distance as a divergence in Definition \ref{def:div} defines a {\it diamond-distance of channels},
\begin{equation}\label{def:qre-channel}
\|\cN - \cM\|_\diamond=\max_{\rho} \|\cN\otimes I(\rho)- \cM\otimes I(\rho)\|_1\ .
\end{equation}
\end{definition}

Recall that for the Umegaki relative entropy the Pinsker's inequality for states holds: for any  states $\rho, \sigma$, we have
$$D(\rho\|\sigma)\geq \frac{1}{2}\|\rho-\sigma\|_1^2\ . $$

Straight from the definition of the quantum relative entropy of channels, the definition of the diamond distance (\ref{def:qre-channel}), and Pinsker’s inequality for states, we obtain Pinsker’s inequality for channels.
\begin{proposition}
For any quantum channels $\cN, \cM$, the Pinsker's inequality for channels holds
$$D_{re}(\cN\|\cM)\geq \frac{1}{2}\|\cN-\cM\|_\diamond^2\geq \frac{1}{2}\|\cN-\cM\|_1^2\ . $$
Here $D_{re}$ is the relative entropy of channels based on the Umegaki relative entropy.
\end{proposition}

\subsection{Static resource theory of coherence}

Let $\cH$ be a $d$-dimensional Hilbert space. Let us fix an orthonormal basis $\cE=\{\ket{j}\}_{j=1}^d$ of vectors in $\cH$.
\begin{definition} A state $\delta$ is called {\it incoherent} in $\cE$ if it can be represented as follows
$\delta=\sum_j \delta_j\ket{j}\bra{j}. $ Denote the set of {\bf incoherent states} for a fixed basis $\cE=\{\ket{j}\}_j$ as $\cI_\cE=\{\rho=\sum_jp_j\ket{j}\bra{j}\}.$ We drop the subscript $\cE$ from now on.
\end{definition}

\begin{notation}\label{notation:cdo}
The {\it completely dephasing channel} in basis $\cE$,  $\Delta_\cE:\mathcal L(H)\to\mathcal L(H)$ is defined as
\begin{equation}\label{def:dephasing}
\Delta_\cE(\rho)=\sum_j \bra{j}\rho\ket{j} \ket{j}\bra{j}\ .
\end{equation}
Sometimes we drop the subscript $\cE$, or use a subscript denoting the system this channel acts on.
\end{notation}

Any quantum channel $\cN$ admits a Kraus decomposition (which may not be unique) of the form
\begin{equation}\label{eq:Kraus}
\cN(\rho)=\sum_n K_n \rho K_n^*\ , \ \ \sum_n K_n^*K_n=I\ .
\end{equation}

\begin{definition} 
A quantum channel $\cN:\cL(\cH)\rightarrow \cL(\cH)$ is an {\it incoherent operation (IO)}, if there exists a Kraus decomposition (\ref{eq:Kraus}) such that  $K_n \cI K_n^*\subset \cI,\ \text{for all }n.$
\iffalse
\begin{itemize}
\item a {\it maximal incoherent operation (MIO)}, if $\cN( \cI) \subset \cI.$
\item an {\it incoherent operation (IO)}, if there exists a Kraus decomposition (\ref{eq:Kraus}) such that  $K_n \cI K_n^*\subset \cI,\ \text{for all }n.$
\item a {\it dephasing-covariant incoherent operation (DIO)}, if $\cN\Delta=\Delta\cN.$
\item  a {\it strictly incoherent operation (SIO)}, if if there exists a Kraus decomposition
(\ref{eq:Kraus})  such that
$K_n \Delta(\rho) K_n^*=\Delta\left(K_n\rho K_n^* \right),\ \text{for all }n.$ 
\end{itemize}
\fi
\end{definition}

\begin{definition}
A {\it measure of coherence} $\cC(\rho)$ should satisfy the following conditions
\begin{itemize}
\item (C1) $\cC(\rho)\geq 0$, and $\cC(\rho)=0$ if and only if $\rho\in\cI$;
\item (C2) Non-selective monotonicity under IO (monotonicity): for all IO $\cN$ and all states $\rho$,
$$\cC(\rho)\geq \cC(\cN(\rho))\ ; $$
\item (C3) Selective monotonicity under IO (strong monotonicity): for all IO $\cN$ with Kraus operators $K_n$, and all states $\rho$,
$$\cC(\rho)\geq \sum_n p_n \cC(\rho_n)\ , $$
where $p_n$ and $\rho_n$ are the outcomes and post-measurement states
$$\rho_n=\frac{K_n\rho K_n^*}{p_n},\ \ p_n=\Tr K_n\rho K_n^*\ . $$
\item (C4) Convexity, 
$$\sum_n p_n \cC(\rho_n)\geq \cC\left(\sum_n p_n\rho_n\right)\ , $$
for any family of states $\{\rho_n\}$ and any probability distribution $\{p_n\}$.
\end{itemize}
\end{definition}
These properties are analogous to the entanglement measure theory, where the average entanglement is not increased under the local operations and classical communication (LOCC). Notice that coherence measures that satisfy conditions (C3) and (C4) also satisfy condition (C2) \cite{BC14}. 

While properties (C2) and (C3) are formulated for incoherent operations, one may consider different classes of incoherent operations (such as maximally incoherent operations, strictly incoherent operations, and many others, see for example  \cite{ChG16, ChG16-2, DevS16} for their comparison and analysis).

\subsection{Dynamical resource theory of coherence}

Quantum channels now play the role of states in the static resource theory. While there is only one commonly agreed set of free incoherent states (diagonal in the pre-fixed basis), there are multiple ways one can consider free channels in the dynamical resource theory, such as classical channels, or detection/creation/detection-creation incoherent channels \cite{SCG20, TEZP19}. 

Let $\cN:A\rightarrow B$ be a quantum channel, and let us fix bases $\cE_A$ and $\cE_B$ in systems $A$ and $B$ respectively. Then $\Delta_A$ and $\Delta_{B}$ denote the completely dephasing channels (\ref{def:dephasing}) in systems $A$ and $B$ with respect to the incoherent bases $\cE_A$ and $\cE_B$ respectively. For a composite system $AB$, we take the reference basis to be the product basis $\cE_A\otimes \cE_B$. Thus the corresponding completely dephasing channel is $$
\Delta_{AB}=\Delta_{A}\otimes\Delta_{B}.
$$

\begin{definition}\label{def:DCI}\cite{SCG20, TEZP19}
A quantum channel $\cN_{A\rightarrow B}$ is
\begin{itemize}
\item {\it detection incoherent (DI)} iff $\Delta_{B}\cN_{A\rightarrow B}=\Delta_{B}\cN_{A\rightarrow B}\Delta_{A},$ 
\item {\it creation incoherent (CI)} iff $\cN_{A\rightarrow B}\Delta_{A}=\Delta_{B}\cN_{A\rightarrow B}\Delta_{A},$
\item {\it detection-creation incoherent (DCI) or DIO} iff $\Delta_{B}\cN_{A\rightarrow B}=\cN_{A\rightarrow B}\Delta_{A}$.
\end{itemize}
See \cite{TEZP19} for more detailed information on the structure of these channels.
\end{definition}

A set of free superchannels can be defined in multiple ways \cite{SCG20, TEZP19, Y18}. 
Similar to \cite{TEZP19, Y18}, we consider the following free superchannels. 

\begin{notation}
A superchannel $\Lambda$ is free if there is a decomposition 
\begin{equation}\label{eq:free-super-one}
\Lambda(\cN_{A\rightarrow B})=\cK_{BE\rightarrow D}\circ(\cN_{A\rightarrow B}\otimes I_{E})\circ\cM_{C\rightarrow AE}\ ,
\end{equation}
where $\cM_{C\rightarrow AE}, \cK_{BE\rightarrow D}$ are free channels taken to be either DI, CI, or DCI.
\end{notation}

\begin{definition}\label{def:coh-meas-dyn} \cite{TEZP19, Y18}
{\it Coherence measures} are functions $C$ of channels $\cN$ satisfying the following conditions:
\begin{itemize}
\item (C1) {Non-negativity:} $C(\cN)\geq 0$ and $C(\cN)=0$ for free channels $\cN$.
\item (C2) {Monotonicity:} for free superchannels $\Lambda$, the coherence measure is non-increasing, $C(\Lambda(\cN))\leq C(\cN)$.
\end{itemize}
Sometimes other conditions are included as well, for example:
\begin{itemize}
\item (C3) Convexity: for channels $\cN, \cK$ and $0\leq \lambda\leq 1$, it holds $C(\lambda \cN+(1-\lambda)\cK)\leq \lambda C(\cN)+(1-\lambda) C(\cK).$
\end{itemize}
\end{definition}

\section{Coherence measures of channels}\label{sec:Coh-meas}

\begin{notation}\label{def:f}
Let a function $f(\cN,\cM)$ between two channels $\cN, \cM$, viewed as a distance-like function between them, be such that 
\begin{itemize}
\item (Non-negativity and separation) For quantum channels $\cN, \cM$,  it holds that $f(\cN,\cM)\geq 0$, and $f(\cN,\cM)=0$ if and only if $\cN=\cM$.
\item (Weak monotonicity) For quantum channels $\cV,\cU,\cN, \cM$, the function is non-increasing under composition $f(\cV\circ\cN\circ\cU, \cV\circ\cM\circ\cU)\leq f(\cN, \cM)$. %It is possible to restrict this condition for when $\cV, \cU$ are free (DI, CI, or DCI) channels only.
\item (Joint convexity). For $0\leq \lambda\leq 1$, it holds that $f\left(\lambda\cN_1+(1-\lambda)\cN_2, \lambda\cM_1+(1-\lambda)\cM_2\right)\leq \lambda f(\cN_1, \cM_1)+(1-\lambda)f(\cN_2,\cM_2)$.
\item (Monotone under tensor product with the identity channel). For quantum channels $\cN, \cM$,  it holds that $f(\cN\otimes I,\cM\otimes I)\leq f(\cN, \cM)$.
\end{itemize}
\end{notation}

%The above conditions are chosen to correspond to the properties required of the resulting coherence measures. Non-negativity and vanishing on identical channels are used to establish condition (C1), the separation property is used below to obtain faithfulness when $f$ is additionally lower semicontinuous, weak monotonicity under composition is used to establish monotonicity under free pre- and post-processing, and joint convexity is used to establish convexity. Monotonicity under tensoring with the identity channel is needed to handle the ancillary identity channel appearing in a free superchannel.

Notice that the above conditions also imply monotonicity under tensor product with a completely dephasing channel. Indeed, for quantum channels $\cN, \cM: A\rightarrow B$, and the completely dephasing channel $\Delta_C$ on any auxiliary system $C$, it holds that 
\begin{align}
f(\mathcal N\otimes\Delta_C,\cM\otimes\Delta_C)
&=
f((I_B\otimes\Delta_C)\circ(\mathcal N\otimes I_C),
(I_B\otimes\Delta_C)\circ(\cM\otimes I_C))\\
&\le f(\mathcal N\otimes I_C,\cM\otimes I_C)
\le f(\mathcal N,\cM). \label{eq:mono-deph}
\end{align}
Here the first inequality follows from the weak monotonicity and the second inequality follows from the monotonicity under tensor product with the identity channel.
This property will be used below to establish monotonicity of the DI and CI coherence measures under the ancillary identity channel appearing in a free superchannel.

The following lemma follows from \cite{Y18}.
\begin{lemma}\label{lemma:RE}
Quantum divergence of channels (see Definition \ref{def:div}), $f(\cN,\cM)=D(\cN\|\cM)$, satisfies the conditions in Notation \ref{def:f} for any divergence from Definition \ref{def:div-div}. In particular, the relative entropy of channels and the diamond distance of channels are such functions.
\end{lemma}

Having established that quantum divergences satisfy the required assumptions, we now state the main result, showing that any function $f$ satisfying Notation \ref{def:f} generates coherence measures for the three classes of free operations.

\begin{theorem}\label{thm:main-f}
Let $f$ be a function satisfying the properties in Notation \ref{def:f}. Consider the three cases in which the sets DI, CI, and DCI (see Definition \ref{def:DCI}) are taken as the sets of free operations. Define:
\begin{align}
&C_f^{DI}(\cN)=\inf_{\cK\in DI}f(\Delta\cN,\Delta\cK),\label{def:DI-f}\\
&C_f^{CI}(\cN)=\inf_{\cK\in CI}f(\cN\Delta,\cK\Delta),\label{def:CI-f}\\
&C_f^{DCI}(\cN)=\inf_{\cK\in DCI}f(\cN,\cK). \label{def:DCI-f}
\end{align}

These functionals are coherence measures of channels, i.e. they satisfy conditions (C1) and (C2) in Definition \ref{def:coh-meas-dyn} and are also convex as in (C3). The free superchannels are of the type (\ref{eq:free-super-one}) with free channels taken from a corresponding set of free operations (DI, CI, or DCI). In other words,
\begin{enumerate}[label=\Roman*.]
\item $C_f^{\mathcal C}(\mathcal N)\ge0$ for every $\mathcal N,$ and $C_f^{\mathcal C}(\mathcal N)=0$ for $\mathcal N\in\mathcal C.$
\item Let $\Lambda$ be a free superchannel of the type (\ref{eq:free-super-one})  with free channels taken from $\cC$. Then the coherence measure is monotone under $\Lambda$,
$$C_f^\cC(\Lambda(\cN))\leq C_f^\cC(\cN)\ .$$
\item For any quantum channels  $\cN, \cK$  and  $0\leq\lambda\leq 1$, the coherence measure is convex
$$C_f^\cC(\lambda \cN+(1-\lambda)\cK)\leq \lambda C_f^\cC(\cN)+(1-\lambda)C_f^\cC(\cK)\ . $$
\end{enumerate}
If, in addition, $f$ is lower semicontinuous, then I can be strengthened to
$$
C_f^{ \cC}( \cN)=0
\quad\Longleftrightarrow\quad
 \cN\in \cC.
$$
\end{theorem}

\begin{remark}
For the DI and CI coherence measures (\ref{def:DI-f}) and (\ref{def:CI-f}), one may try to interpret the definitions as distances to the corresponding free sets. For example, in the DI case one can define 
$$
F(\mathcal N,\cM)=f(\Delta\mathcal N,\Delta\cM).
$$
Although $\mathcal N=\cM$ implies $F(\mathcal N,\cM)=0$, the converse does not hold in general: distinct channels $\mathcal N\neq\cM$ may satisfy $F(\mathcal N,\cM)=0$ whenever
$$
\Delta\mathcal N=\Delta\cM.
$$
Thus $F$ is not a faithful distance on the full set of channels. Similarly, in the CI case, for the analogous function $F(\mathcal N,\cM)=f(\mathcal N\Delta,\cM\Delta)$, distinct channels may have zero $F$-distance whenever
$$
\mathcal N\Delta=\cM\Delta.
$$
Thus, the DI and CI constructions cannot in general be interpreted as ordinary distances to the free set with respect to a faithful distance on the full set of quantum channels. This does not affect Theorem \ref{thm:main-f}, since Definition \ref{def:coh-meas-dyn} requires a coherence measure to vanish on free channels, but does not require the converse. Accordingly, for a general $f$ satisfying Notation \ref{def:f}, faithfulness need not follow from the assumptions alone; as shown in the proof of Theorem \ref{thm:main-f}, it does hold when $f$ is additionally lower semicontinuous. 

We note that the trace-one distance between channels does not satisfy the assumptions of Notation \ref{def:f}, since it is not in general monotone under tensoring with an identity channel. Nevertheless, \cite{TEZP19} shows directly that the corresponding DI functional (\ref{def:DI-f}) is a coherence measure. Their proof uses the dephasing appearing in the DI construction, which also dephases the ancillary system after tensoring with an identity channel, and therefore does not imply the required tensor-product property for the trace-one distance itself. A similar direct argument appears to apply to the CI construction, although we do not pursue this case here.

\end{remark}

\begin{proof} 
Denote the set of free operations DI, CI, or DCI as $\cC$.

I. Non-negativity follows immediately from the non-negativity of \(f\). If $\cN\in\cC$, then $C_f^\cC(\cN)=0$.

If, in addition, $f$ is lower semicontinuous, then the converse also holds, so that $C_f^{\mathcal C}(\mathcal N)=0$ if and only if $\mathcal N\in\mathcal C$. This includes, in particular, the diamond distance and the relative entropy of channels. Indeed, since the set of quantum channels is compact in finite dimensions and each of DI, CI, and DCI is closed (being defined by linear equalities), the infimum in (\ref{def:DI-f})–(\ref{def:DCI-f}) is attained.  Therefore, if $C_f^\cC(\cN)=0$, then there exists $\cK\in\cC$ attaining the infimum. By the separation property of $f$ in Notation \ref{def:f}, the vanishing of the corresponding value of $f$ implies:
\begin{itemize}
\item for $\cC=DI$: $\Delta\cN=\Delta\cK=\Delta\cK\Delta=\Delta\cN\Delta$, therefore $\cN\in DI$.
\item for $\cC=CI$: $\cN\Delta=\cK\Delta=\Delta\cK\Delta=\Delta\cN\Delta$, therefore $\cN\in CI$.
\item for $\cC=DCI$: $\cN=\cK\in DCI$.
\end{itemize}

II. To prove monotonicity, let us note a few things.\\
1. A composition of two maps in $\cC$ belongs to $\cC$. \\
a) Suppose
$\cK_2:A\to B,$ and $\cK_1:B\to C,$
with $\cK_1, \cK_2\in DI$. Then $\cK_1\cK_2\in DI$ since
\begin{align}
\Delta_C\cK_1\cK_2&=\Delta_C\cK_1\Delta_B\cK_2\\
&=\Delta_C\cK_1\Delta_B\cK_2\Delta_A\\
&=\Delta_C\cK_1\cK_2\Delta_A\ .
\end{align}
The first equality holds since $\cK_1\in DI$, i.e. $\Delta_C\cK_1=\Delta_C\cK_1\Delta_B$. The second equality holds since $\cK_2\in DI$, i.e. $\Delta_B\cK_2=\Delta_B\cK_2\Delta_A$. The last equality holds since again $\cK_1\in DI$ and therefore $\Delta_C\cK_1\Delta_B=\Delta_C\cK_1$.

b) Similarly, if $\cK_1, \cK_2\in CI$, then $\cK_1\cK_2\in CI$ since
\begin{align}
\cK_1\cK_2\Delta_A&=\cK_1\Delta_B\cK_2\Delta_A\\
&=\Delta_C\cK_1\Delta_B\cK_2\Delta_A\\
&=\Delta_C\cK_1\cK_2\Delta_A\ .
\end{align}
The first equality holds since $\cK_2\in CI$, i.e. $\cK_2\Delta_A=\Delta_B\cK_2\Delta_A$. The second equality holds since $\cK_1\in CI$, i.e. $\cK_1\Delta_B=\Delta_C\cK_1\Delta_B$. The last equality holds since again $\cK_2\in CI$ and therefore $\Delta_B\cK_2\Delta_A=\cK_2\Delta_A$.

c) Similarly, if $\cK_1, \cK_2\in DCI$, then $\cK_1\cK_2\in DCI$ since
\begin{align}
\Delta_C\cK_1\cK_2&=\cK_1\Delta_B\cK_2\\
&=\cK_1\cK_2\Delta_A\ .
\end{align}
The first equality holds since $\cK_1\in DCI$, i.e. $\Delta_C\cK_1=\cK_1\Delta_B$. The second equality holds since $\cK_2\in DCI$, i.e. $\Delta_B\cK_2=\cK_2\Delta_A$.

2. For any free channel $\Phi\in\cC$, we have $C_f^\cC(\cN_{A\rightarrow B}\circ\Phi_{E\rightarrow A})\leq C_f^\cC(\cN_{A\rightarrow B})$.\\
a) Let $\Phi\in DI$ be any channel in DI. Then
\begin{align}
C_f^{DI}(\cN_{A\rightarrow B}\circ\Phi_{E\rightarrow A})&=\inf_{\cK_{E\rightarrow B}\in DI}f(\Delta_B\cN_{A\rightarrow B}\Phi_{E\rightarrow A},\Delta_B\cK_{E\rightarrow B})\\
&\leq\inf_{\cK'_{A\rightarrow B}\in DI}f(\Delta_B\cN_{A\rightarrow B}\Phi_{E\rightarrow A},\Delta_B\cK'_{A\rightarrow B}\Phi_{E\rightarrow A})\\
&\leq\inf_{\cK'_{A\rightarrow B}\in DI}f(\Delta_B\cN_{A\rightarrow B},\Delta_B\cK'_{A\rightarrow B})\\
&= C_f^{DI}(\cN_{A\rightarrow B})\ .
\end{align}
In the first inequality we used that the composition of DI maps is DI. In the second inequality we used that the function $f$ is weakly monotone for quantum channels.

b) Let $\Phi\in CI$ be any channel in CI. Then
\begin{align}
C_f^{CI}(\cN\circ \Phi)&=\inf_{\cK\in CI}f(\cN\Phi\Delta,\cK\Delta)\\
&\leq\inf_{\cK'\in CI}f(\cN\Phi\Delta,\cK'\Phi\Delta)\\
&=\inf_{\cK'\in CI}f(\cN\Delta\Phi\Delta,\cK'\Delta\Phi\Delta)\\
&\leq\inf_{\cK'\in CI}f(\cN\Delta,\cK'\Delta)\\
&= C_f^{CI}(\cN)\ .
\end{align}
In the first inequality we used that the composition of CI maps is CI. The last inequality is true, since the function $f$ is weakly monotone on quantum channels.

c) Let $\Phi\in DCI$ be any channel in DCI. Then
\begin{align}
C_f^{DCI}(\cN\circ \Phi)&=\inf_{\cK\in DCI}f(\cN\Phi,\cK)\\
&\leq\inf_{\cK'\in DCI}f(\cN\Phi,\cK'\Phi)\\
&\leq\inf_{\cK'\in DCI}f(\cN,\cK')\\
&= C_f^{DCI}(\cN)\ .
\end{align}
In the first inequality we used that the composition of DCI maps is DCI. The last inequality is true, since the function $f$ is weakly monotone on quantum channels.

3. For any free channel $\Phi\in\cC$, we have $C_f^\cC(\Phi\circ\cN)\leq C_f^\cC(\cN)$.\\
a)  Let $\Phi\in DI$. Then
\begin{align}
C_f^{DI}(\Phi\circ \cN)&=\inf_{\cK\in DI}f(\Delta\Phi\cN,\Delta\cK)\\
&\leq \inf_{\cK'\in DI}f(\Delta\Phi\cN,\Delta\Phi\cK')\\
&= \inf_{\cK'\in DI}f(\Delta\Phi\Delta\cN,\Delta\Phi\Delta\cK')\\
&\leq \inf_{\cK'\in DI}f(\Delta\cN,\Delta\cK')\\
&=C_f^{DI}(\cN)\ .
\end{align} 
The last inequality is true, since $f$ is weakly monotone.

b) Let $\Phi\in CI$. Then
\begin{align}
C_f^{CI}(\Phi\circ \cN)&=\inf_{\cK\in CI}f(\Phi\cN\Delta,\cK\Delta)\\
&\leq \inf_{\cK'\in CI}f(\Phi\cN\Delta,\Phi\cK'\Delta)\\
&\leq \inf_{\cK'\in CI}f(\cN{\Delta},\cK'{\Delta})\\
&=C_f^{CI}(\cN)\ .
\end{align} 
The last inequality is true, since $f$  is weakly monotone.

c) Let $\Phi\in DCI$. Then
\begin{align}
C_f^{DCI}(\Phi\circ \cN)&=\inf_{\cK\in DCI}f(\Phi\cN,\cK)\\
&\leq \inf_{\cK'\in DCI}f(\Phi\cN,\Phi\cK')\\
&\leq \inf_{\cK'\in DCI}f(\cN,\cK')\\
&=C_f^{DCI}(\cN)\ .
\end{align} 
The last inequality is true, since $f$ is weakly monotone.

4. We also use that if $\mathcal \cK:A\to B$ belongs to any of the three free classes, then $\mathcal \cK\otimes I_E:AE\to BE$ belongs to the same class. \\
a) If $\cK\in DI$, then 
\begin{align}
(\Delta_B\otimes\Delta_E)(\cK\otimes I_E)
&=\Delta_B\cK\otimes\Delta_E\\
&=\Delta_B\cK\Delta_A\otimes\Delta_E\\
&=(\Delta_B\otimes\Delta_E)
   (\cK\otimes I_E)
   (\Delta_A\otimes\Delta_E)
\end{align}

b) If $\cK\in CI$, then 
\begin{align}
(\cK\otimes I_E)(\Delta_A\otimes\Delta_E)
&=\cK\Delta_A\otimes\Delta_E\\
&=\Delta_B\cK\Delta_A\otimes\Delta_E\\
&=(\Delta_B\otimes\Delta_E)
   (\cK\otimes I_E)
   (\Delta_A\otimes\Delta_E)
\end{align}

c) If $\cK\in DCI$, then 
\begin{align}
(\Delta_B\otimes\Delta_E)(\cK\otimes I_E)
&=\Delta_B\cK\otimes\Delta_E\\
&=\cK\Delta_A\otimes\Delta_E\\
&=(\cK\otimes I_E)(\Delta_A\otimes\Delta_E)
\end{align}

5. It also holds that $C_f^\cC(\cN_{A\rightarrow B}\otimes I_E)\leq C_f^\cC(\cN_{A\rightarrow B})$.\\
a) Consider
\begin{align}
C_f^{DI}(\cN_{A\rightarrow B}\otimes I_E)&=\inf_{\cK_{AE\rightarrow BE}\in DI}f(\Delta_{BE}(\cN_{A\rightarrow B}\otimes I_E),\Delta_{BE}\cK_{AE\rightarrow BE})\\
&\leq \inf_{\cK'_{A\rightarrow B}\otimes I_E\in DI}f(\Delta_{BE}(\cN_{A\rightarrow B}\otimes I_E),\Delta_{BE}(\cK'_{A\rightarrow B}\otimes I_E))\\
&=\inf_{\cK'_{A\rightarrow B}\in DI}f(\Delta_B\cN_{A\rightarrow B}\otimes \Delta_E,\Delta_B\cK'_{A\rightarrow B}\otimes \Delta_E)\\
&\leq \inf_{\cK'_{A\rightarrow B}\in DI}f(\Delta_B\cN_{A\rightarrow B},\Delta_B\cK'_{A\rightarrow B})\\
&=C_f^{DI}(\cN_{A\rightarrow B})\ .
\end{align}
The first inequality holds by point 4. above. The last inequality holds since $f$ is monotone under tensor product with the dephasing channel (\ref{eq:mono-deph}).

b) Similarly,
\begin{align}
C_f^{CI}(\cN\otimes I)&=\inf_{\cK\in CI}f((\cN\otimes I)\Delta,\cK\Delta)\\
&\leq \inf_{\cK=\cK'\otimes I\in CI}f((\cN\otimes I)\Delta,(\cK'\otimes I)\Delta)\\
&=\inf_{\cK'\in CI}f(\cN\Delta\otimes \Delta,\cK'\Delta\otimes \Delta)\\
&\leq \inf_{\cK'\in CI}f(\cN\Delta,\cK'\Delta)\\
&=C_f^{CI}(\cN)\ .
\end{align}

c) Also,
\begin{align}
C_f^{DCI}(\cN\otimes I)&=\inf_{\cK\in DCI}f(\cN\otimes I,\cK)\\
&\leq \inf_{\cK=\cK'\otimes I\in DCI}f(\cN\otimes I,\cK'\otimes I)\\
&\leq \inf_{\cK'\in DCI}f(\cN,\cK')\\
&=C_f^{DCI}(\cN)\ .
\end{align}
In the last inequality we used that $f$ is monotone under tensor product with the identity channel.

From steps 1-5, for any free quantum superchannel $\Lambda(\cN_{A\rightarrow B})=\Phi_{BE\rightarrow D}\circ(\cN_{A\rightarrow B}\otimes I_{E})\circ\Psi_{C\rightarrow AE}$ with free $\Phi, \Psi\in \cC$, all coherence measures are monotone,
$$C_f^\cC(\Lambda(\cN))\leq C_f^\cC(\cN)\ . $$

III. Let $0< \lambda< 1$. The cases $\lambda=0,1$ are trivial.
Then since the sets DI, CI, and DCI are convex, we obtain
\begin{align}
C_f^{DI}(\lambda\cN+(1-\lambda)\cK )&=\inf_{\Phi\in DI}f\left(\Delta(\lambda\cN+(1-\lambda)\cK ),\Delta\Phi\right)\\
&\leq\inf_{\cK_1, \cK_2\in DI} f\left(\Delta(\lambda\cN+(1-\lambda)\cK ),\Delta(\lambda\cK_1+(1-\lambda)\cK_2)\right)\\
&= \inf_{\cK_1, \cK_2\in DI} f\left(\lambda \Delta\cN+(1-\lambda)\Delta\cK ,\lambda\Delta\cK_1+(1-\lambda)\Delta\cK_2\right)\\
&\leq\inf_{\cK_1, \cK_2\in DI} \left[ \lambda f(\Delta\cN, \Delta\cK_1)+(1-\lambda)f(\Delta\cK,\Delta\cK_2)\right]\\
&=\lambda C_f^{DI}(\cN)+(1-\lambda)C_f^{DI}(\cK)\ .
\end{align}
Here, in the second inequality we used the joint convexity of $f$.

{Similar results are straightforward for CI and DCI.}
\end{proof}

\vspace{0.3in}
\textbf{Acknowledgments.}  A. V. is supported by NSF grant DMS-2105583.
\vspace{0.3in}

\bibliography{Bibliography-Vershynina}{}
\bibliographystyle{plain}

\end{document}